\documentclass[11pt]{article}

\usepackage[tikz]{KLMM/klmm}

\newcommand{\ZZ}{\mathbb Z}
\newcommand{\Z}{\mathbb Z}

\def\Q{{\mathbb{Q}}}
\def\Z{{\mathbb{Z}}}
\def\N{{\mathbb{N}}}

\def\F{{\mathbb{F}}}

\def\Pr{{\mathbb{P}}}

\def\Size{\text{Size}}

\newcommand{\Fp}{\mathbb F_p}

\newcommand{\PP}{\mathsf P}
\newcommand{\NP}{\mathsf{NP}}
\DeclareMathOperator{\ord}{ord}

\DeclareMathOperator{\Res}{Res}
\DeclareMathOperator{\cont}{cont}
\DeclareMathOperator{\poly}{poly}
\DeclareMathOperator{\supp}{supp}
\DeclareMathOperator{\size}{size}

\title{Sparse Polynomial GCD Cannot Be Computed in Fixed-output Polynomial Time}
\author{\fontsize{11}{16}\selectfont
    Yichuan Cao\institution{ State Key Laboratory of Mathematical Sciences, Academy of Mathematics and Systems Science, Chinese Academy of Sciences; University of Chinese Academy of Sciences}
    \and Ruyong Feng\instref{1}
    \and Xiao-Shan Gao\instref{1}
    \and Qiao-Long Huang\institution{ Shandong University}
    \and  Ruichen Qiu\instref{1}
}
\date{\today}

\begin{document}
\maketitle

\begin{abstract}
\noindent
This paper shows that, assuming $\PP\ne\NP$, sparse univariate polynomial GCD cannot be computed in fixed-output polynomial time over any fixed computable field.
Precisely, there is no deterministic algorithm that can compute the GCD of sparse univariate polynomials over any fixed computable field with total bit complexity polynomial in the input size, with the output size as a fixed parameter. 
In particular, there is no deterministic algorithm for sparse polynomial GCD that runs in polynomial time with respect to both the input and the output sizes. 
This settles the open problem posed as Challenge 5 in \cite{SparsityChallenges}.
%Inputshave binary exponents and finite-string coefficient names; outputs are explicit normalized sparse lists, with a single fixed encoding of the polynomial $1$.
Furthermore, we show that the Roots of Unity Detection problem over finite fields is NP-hard under random reductions; that is, determining whether the GCD of a sparse univariate polynomial and $x^n - 1$ has nonzero degree is NP-hard.
\end{abstract}

\section{Introduction}
\label{sec-0}
The greatest common divisor (GCD) for polynomials is one of the core and most extensively studied problems in computer algebra \cite{ModernCA}.
%\cite{ModernCA,SparsityChallenges,Roche2018}. 
For two dense degree-$D$ univariate polynomials over a field, the fast Euclidean algorithm can be used to compute their GCD using $O(D\,\log^2 D\,\log\log D)$ field operations~\cite[Chapter 11]{ModernCA}. 

In practical computations and in computer algebra systems like Maple and Mathematica, polynomials are typically represented in a sparse form~\cite{SparsityChallenges,Roche2018}, meaning that only the nonzero coefficients and their corresponding exponents are stored in lists.
Consider a sparse univariate polynomial $\F$ of degree $D$, containing $T$ nonzero terms, and whose coefficients have a maximum absolute value of $H$. Its size is tightly bounded by
$$\Size(f) = O(T(\log D+\log H))=O(T\log(DH)).$$
So a sparse polynomial-time algorithm should have bit complexity that is polynomial in $T$, $\log D$, and $\log H$.

It is well known that there exist no sparse polynomial-time GCD algorithms because the GCD of two sparse polynomials can have a size that is exponential in the sizes of the input polynomials, as shown below.
\begin{example}[\cite{GCD-Schinzel2003}]
\label{ex-GCD1}
Let $p$, $q$ be distinct primes and 
$f=x^{pq} - 1, g=x^{p+q} - x^p - x^q + 1=(x^{p}-1)(x^{q}-1)$. Then we have the following irreducible decompositions:
$x^{p}-1 = (x-1)\Phi_{p}$,
$x^{q}-1 = (x-1)\Phi_{q}$,
$x^{pq}-1 = (x-1)\Phi_{p}\Phi_{q}\Phi_{pq}$, 
where $\Phi_{n}$ is the $n$-th cyclotomic polynomial.
Thus
\begin{align*}
\gcd(f,g) 
&= (x-1)\Phi_{p}\Phi_{q} =(x^p-1)\Phi_{q}= x^p\Phi_{q} - \Phi_{q}
\end{align*}
which contains $2q$ terms, assuming $q<p$.
Thus, the size of the GCD is $O(q)$ which is exponential in the input size $O(\log q)$
(assuming $p = q+ c$ for some constant $c$).
\end{example}

This implies that the strongest guaranty we can provide is output-sensitive complexity; in other words, the algorithm’s running time must explicitly depend on the size of the GCD. 

There exist many breakthrough algorithms for GCD, both output-sensitive and output-insensitive, 
including methods based on sub-resultant computation \cite{collins1967subresultants,brown1971subresultants,hong2023computing}, 
methods based on interpolations \cite{zippel1979probabilistic,van2021sparse,hu2016fast,GCDHuangGao2025,kaltofen1990computing}, and modular methods based on Hensel lifting \cite{Moses1973EZ,huang2026sparse}. 
Recent research has continued to advance GCD computation, focusing on improvements such as the half-GCD algorithm \cite{van2025optimizing}. Additionally, constant-time variants of the GCD algorithm have been developed, which are particularly valuable in cryptographic applications where resistance to timing attacks is essential \cite{bernstein2019fast}.
In many cases, these algorithms perform quite effectively, as demonstrated by their use in popular software systems such as Maple and Mathematica.
However, the complexity of all these methods depends on the degrees of the input polynomials; therefore, they are not polynomial-time algorithms.

The following example shows that, even for output-sensitive complexity, the procedure to compute the GCD might be exponential.
\begin{example}
\label{ex-GCD2}
Let $p>q$ be two distinct primes, and 
$f=x^{p} - 1, g=x^{q} - 1$. 
%Then we have the following irreducible decomposition $x^{p}-1 = (x-1)\Phi_{p}$ and  $x^{q}-1 = (x-1)\Phi_{q}$. 
Then their GCD can be uniquely written as
\begin{align*}
%&x^{p}-1 = (x-1)\Phi_{p}\\
%&x^{q}-1 = (x-1)\Phi_{q}\\
&\gcd(f,g) = x-1 = Af+Bg.
\end{align*}
We can show that $T_B>p-q$; that is, even if $T_{\gcd(f,g)}$ is small, computing $\gcd(f,g)$ might require handling polynomials $A$ and $B$ that are exponential in the sizes of $f, g,\gcd(f,g)$.
\end{example}
 
%\textbf{Research problem 16.17 \cite{ModernCA}} 
%Can one compute the GCD of two multivariate polynomials in random polynomial time in the length of the sparse representation plus the degree? Is the output length always polynomial in the input length? (If not, one might consider the combined input plus output lengths in the first question.)

%In \cite[Theorem 3.3]{plaisted1984}, it is proven that determining whether the GCD of two sparse integer polynomials has a positive degree is an NP-hard problem.

In the seminal work, Plaisted proved that determining whether two sparse univariate polynomials over $\Q$ possess a non-trivial GCD (referred to as {\em noncoprimality} in this paper) is NP-hard~\cite[Theorem 5.1]{plaisted1984}. 
%Noncoprimality of two sparse polynomials over a certain finite field is shown to be NP-hard under random reductions~\cite{Gathen1993CC,kaltofen2005complexity}.
The NP-hardness of noncoprimality is shown via a randomized reduction for the finite field $\F_{p^m}$, for any choice of $p$ or when $p^m$ is sufficiently large.
This brings us directly to the following problem~\citep[Challenge 5]{SparsityChallenges} about the output-sensitive sparse polynomial GCD computation.

\begin{quote}
\textbf{Open Problem 1.} 
Find an algorithm for computing $\gcd(f,g)$ that is polynomial-time in $\|f\|_0$, $\|g\|_0$, and $\|\gcd(f,g)\|_0$.
\end{quote}

%For finite fields, we make Open Problem 1 more explicit as follows.
%
%\begin{quote}
%\textbf{Open Problem 2.} 
%For $f,g\in\F_q[x]$, if there exists an algorithm for computing $\gcd(f,g)$ that is polynomial in $T$, $\log D$, and $\log q$, where 
%$$D = \max(\deg(f),\deg(g)),  T = \max(\|f\|_0, \|g\|_0, \|\gcd(f,g)\|_0).$$
%\end{quote}
%Here, GCD is assumed to be monic  to ensure its uniqueness.

\subsection{Main Results}
In this paper, we settle Open Problem 1 by proving that output-sensitive sparse polynomial GCD computation is NP-hard and that the sparse GCD cannot be computed in fixed-output polynomial time.
First, we have the following conditional obstruction.  

\begin{theorem}\label{thm-m1}
Assume $\PP\ne\NP$. For every fixed computable field $\F$, there is no deterministic algorithm that, for all sparse $f,g\in \F[X]$, computes $\gcd(f,g)$ in worst-case
total bit time bounded by a polynomial in the sum of the input 
and output sizes. 
%Both the algorithm and the putative polynomial bound are allowed to depend on $\F$.
%Furthermore, there is no deterministic algorithm that computes $\gcd(f,g)$ and runs in fixed-output polynomial time, a regime strictly larger than output-sensitive polynomial time.
\end{theorem}

%As an intermediate step in proving Theorem \ref{thm-m1}, we showed that deciding noncoprimality over any finite field is NP-hard. 
%The hardness result for noncoprimality in \cite{kaltofen2005complexity} is under random reduction and for a finite field $\F_{p^m}$ with arbitrary $p$ or sufficiently large $p^m$.

%\begin{theorem}
%\label{thm-m1}
%Let $\F$ be a computable field. If there were an algorithm that, for every input pair \(f,g\in{\mathbb F}[x]\) satisfying
%\[  \deg (f),\deg (g)\le D,\qquad
%  |\operatorname{supp}(f)|,\ |\operatorname{supp}(g)|,\
%  |\operatorname{supp}(\gcd(f,g))|\le T,
%\]
%computes the sparse GCD in bit complexity that is polynomial in \(T\), \(\log D\), and \(\log q\), then \(\mathrm{NP}\subseteq\mathrm{BPP}\). Thus, under the standard complexity assumption \(\mathrm{NP}\nsubseteq\mathrm{BPP}\), no such algorithm exists.
%\end{theorem}

The NP-hardness of output-sensitive GCD computation stems from the NP-hardness of noncoprimality, which was established over integers  \cite{plaisted1984} and over certain finite fields under random reductions  \cite{kaltofen2005complexity}. 
So, the main technical novelty of this paper is the NP-hardness of noncoprimality  over arbitrary finite fields under deterministic reduction.

The above proof procedure leads to stronger complexity results using the parameterized complexity language (Refer to Theorem \ref{thm-strengthened}).

\begin{theorem}\label{thm-m11}
Assume $\PP\ne\NP$. For every fixed computable field $\F$, there is no deterministic algorithm that, for all sparse $f,g\in \F[X]$, computes $\gcd(f,g)$ and runs in fixed-output polynomial time, that is, polynomial in the input size and has the output size as a fixed parameter.
\end{theorem}
%sparse GCD cannot be computed in fixed-output polynomial time  (Theorem \ref{thm-strengthened}) and is not in the fixed-parameter tractable regime for parameter $M$ (Corollary \ref{cor:not-fpt}; $\mathrm{FPT}_{M}$, $M$ being the size of the GCD).
This result is stronger than those in Theorem \ref{thm-m1}.
For example, Theorem \ref{thm-m11} implies that there exists no algorithm with a running time $\poly(N)^M$ that computes $\gcd(f,g)$, where $N$ and $M$ denote the sizes of the input and the output.
These complexity results show that computing the GCD of sparse polynomials is hard, \emph{not} simply that “the algorithm runs slowly when the GCD is large.”  It is ``the algorithm cannot rapidly certify that the GCD is~$1$''.  Every instance produced by the reduction has a constant-size output exactly when it is a \emph{no-instance}, so the entire hardness is already present in the constant-output regime. 
%Larger-$M$ behavior is, for the purpose of the lower bound.

%\begin{remark}
The NP-hardness of noncoprimality leads to the NP-hardness of several problems, including 
testing whether a univariate sparse polynomial is square-free, the ring homomorphism injection problem, 
deciding if two sparse homogeneous linear differential equations in one variable with constant coefficients have a common nontrivial solution,
and testing irreducibility for bi-variate sparse polynomials \cite{plaisted1984,Karpinski1999CH,kaltofen2005complexity}.
As a result of our NP-hardness result for noncoprimality, these problems are NP-hard under deterministic reductions over any field.
In \cite{plaisted1984,Karpinski1999CH,kaltofen2005complexity}, NP-hardness is only established under certain cases.
%\end{remark}

%Open Problem 1 was solved in the first version of this paper over a finite field under random reductions~\cite{qiu2026outputsensitive}, while the result in this paper is deterministic.
%emains unsolved in the context of integers.
We may further ask whether the noncoprimality problem is NP-complete or harder? 
From an algorithmic aspect, can we design a GCD algorithm for two sparse univariate polynomials whose complexity improves on that of the dense GCD algorithm?
%Or is the decision of whether a sparse polynomial is the GCD of two sparse polynomials NP-complete?

\medskip
Furthermore, we show that the \textit{Roots of Unity Detection} problem is NP-hard; that is, given a sparse univariate polynomial \(f\) of degree bounded by \(D\), determining whether \(f\) vanishes at a root of unity whose order is at most \(D\) is NP-hard.
Equivalently, deciding whether \(\deg\gcd(f,g)>0\) is NP-hard for $g=x^n-1$ and for a sparse polynomial $\F$.  

\begin{theorem}
\label{thm-m2}
If there were an algorithm that, for  $f\in\F_q[x]$ and $g=x^n-1$ satisfying
\[
  \deg f\le D\qquad n\le D\qquad
  |\operatorname{supp}(f)|\le T,
\]
decides whether \(\deg\gcd(f,g)>0\) in bit complexity that is polynomial
in \(T\), \(\log D\), and \(\log q\), then \(\mathrm{NP}\subseteq\mathrm{BPP}\).
Thus, under the standard complexity assumption \(\mathrm{NP}\nsubseteq\mathrm{BPP}\), no
such algorithm exists.
\end{theorem}

%\begin{theorem}
%\label{thm-m3}
%For  a sparse polynomial \(f\in\Z[x]\), and $g=x^n-1$, deciding whether \(\deg\gcd(f,g)>0\) is NP-hard.
%\end{theorem}

Theorem \ref{thm-m2} strengthens the previously known result that deciding whether two sparse polynomials are co-prime is NP-hard \cite[Theorem 3.3]{plaisted1984}.
Note that Theorem \ref{thm-m2} is not output-sensitive, so it cannot be deduced from Theorem \ref{thm-m1}, since $\gcd(x^n-1,f)$ could have exponential size, as shown in Example \ref{ex-GCD1}.
Theorem \ref{thm-m2} states that determining whether there exists a unit root that is a zero of $\F$ is NP-hard, whereas \cite[Theorem 4.1]{plaisted1984} shows that deciding whether there exists a unit root that is not a zero of $\F$ is NP-hard. Therefore, Theorem \ref{thm-m2} cannot be inferred from \cite[Theorem 4.1]{plaisted1984}.

%\subsection{Proof Sketch For Theorem \ref{thm-m1} and Technical Novelty}
%\label{sec-setup}

\subsection{Related Work}

%However, the above result does incorporate the term bound of $\gcd(d,g)$, as required in Open Problem 1. 
%If this condition is omitted, computing the GCD of two sparse polynomials is clearly exponential, as illustrated by the following example.
%If the degree $D$ is allowed, there exist polynomial-time algorithm in $T$.

Most of the classic NP-hardness results concerning GCD computation are due to Plaisted \cite{plaisted1977sparse,zippel1979probabilistic,plaisted1984}. These include the NP-hardness of determining whether the GCD of sparse integer polynomials is nontrivial \cite[Theorem 3.3]{plaisted1984}, as well as the NP-completeness of deciding whether $x^N - 1$ fails to divide a given sparse polynomial \cite[Theorem 4.1]{plaisted1984}.
Noncoprimality  was shown to be NP-hard over finite fields under random reductions  \cite{Gathen1993CC,kaltofen2005complexity}.
Deciding the square-freeness of a sparse polynomial in $\Z[x]$ is NP-hard \cite{Karpinski1999CH}.
%The Co-prime test of sparse polynomials over large finite fields is CoNP-hard~\cite{kaltofen2005complexity}.

For univariate polynomials represented by straight-line programs, an algorithm that runs in time polynomial in $\Size(f,g)$ and $D$ is provided in \cite{GCD-Kaltofen1988}.
Filaseta, Granville, and Schinzel give a sparse GCD algorithm that is linear in $\log D$ when one of the polynomials does not have cyclotomic-factors~\cite{FGS2008}.

Multivariate polynomial GCD computation is central to algebraic and symbolic computation.
The mainstream approach reduces the multivariate case to univariate computations via modular methods or Hensel lifting, followed by interpolation or lifting to recover the true GCD. 
Representative contributions include Brown's interpolation-based algorithm for dense polynomials \citep{brown1971euclid}, the EZ-GCD algorithm using Hensel lifting for sparse polynomials by Moses and Yun \citep{Moses1973EZ}, 
Zippel's probabilistic sparse GCD algorithm based on sparse interpolation and the Schwartz--Zippel lemma \citep{zippel1979probabilistic}.
Additionally, Kaltofen's algorithm for polynomials given by straight-line programs \citep{GCD-Kaltofen1988},  extends naturally to the multivariate case. Likewise, the black-box GCD algorithm of Kaltofen and Trager \citep{kaltofen1990computing} also applies to multivariate polynomials.

The parallel sparse GCD algorithm of Hu and Monagan combines Kronecker substitution with Ben-Or--Tiwari sparse interpolation \citep{hu2016fast}. 
These methods typically run in polynomial time, but often rely on randomization or additional assumptions to achieve deterministic complexity or explicit bit complexity bounds.
An algorithm of bit complexity $\widetilde{O}(nDT_G(T_f+T_g)\log^2 q\log^2 \log^2 1/\delta)$ is given over $\F_q$ in \cite{GCDHuangGao2025}.
An algorithm of field complexity $\widetilde{O}(nDT \log H_i \log H_o)$ is given \cite{huang2026sparse}.

This manuscript includes all results from our earlier paper \cite{qiu2026outputsensitive}. 
In particular, Theorems \ref{thm-m1} and \ref{thm-m11} are newly added, and the author list has changed. Therefore, we consider this a new paper.

\medskip
The rest of the paper is organized as follows.
%In Section\ref{sec-pre}, preliminary results are given.
%In Section \ref{sec-charp}, noncoprimality over fields of positive characteristic is presented.
%In Section \ref{sec-charz}, noncoprimality over fields of characteristic zero is presented.
In Section \ref{sec-osgcd}, the proof for Theorem \ref{thm-m1} is presented.
In Section \ref{sec-ur}, the proof for Theorem \ref{thm-m2} is presented.
In Section \ref{sec-conc}, conclusions are presented.

\section{Sparse Polynomial GCD Cannot Be Computed in Fixed-output Polynomial Time}
\label{sec-osgcd}
In this section, we prove Theorems \ref{thm-m1} and \ref{thm-m11}.
\subsection{Proof Sketch}
\label{sec-setup}

In this section, we will prove Theorem \ref{thm-m1}.
The proof uses the standard NP-completeness of 3SAT. We construct, from 3SAT $\Phi$, in deterministic polynomial bit time, two sparse polynomials $f,g$ such that
\begin{equation}\label{eq:reduction}
 \Phi\text{ is satisfiable}\quad\Longleftrightarrow\quad
 \deg\gcd(f,g)>0.
\end{equation}
The first polynomial will always be nonzero, so the alternative to a
positive-degree normalized gcd is exactly $1$. The construction differs
between positive characteristic and characteristic zero, but in both cases
all input coefficients are copied from a fixed finite dictionary of names.
%The last step runs a hypothetical gcd algorithm only long enough to force its answer on the instances with gcd $1$.
%
%We give deterministic reductions from 3SAT in every fixed positive characteristic, including characteristic two, and in characteristic zero.
The positive characteristic case uses Boolean indicators at roots of unity and an auxiliary prime that separates clause values by Frobenius conjugation. 
The latter uses a sum
of squared moduli and separates signed monomial occurrences into distinct
exponent blocks. Both reductions use fixed finite coefficient dictionaries,
so no efficient arithmetic assumption on the field presentation is needed.
A cutoff based only on the output length of $1$ converts an output-sensitive
gcd algorithm into a polynomial-time decision algorithm for 3SAT.

Write $n$ for the number of variables and $m$ for the number of clauses of
$\Phi$. Clauses have at most three literals. We carry out the constructions
for $n,m\ge1$; cases with no variables or no clauses are directly decidable
and can be mapped to $(1,1)$ when unsatisfiable and $(X,X)$ when satisfiable.
An empty clause is allowed and is represented by an empty product, equal to
$1$. Variables may be indexed by their occurrences, so $n$ and $m$ are bounded by the formula's encoded length. Logarithms in bit estimates have base two.

\subsection{Preliminary Results of Primes}
\label{sec-pre}

The reductions need primes whose numerical magnitudes, not merely their bit
lengths, are polynomially bounded. The following elementary estimate permits
trial division throughout.

\begin{lemma}\label{lem:primes}
For every integer $L\ge1$, at least $2L$ primes are at most $32L^2$.
They can be enumerated deterministically in time polynomial in $L$.
\end{lemma}
\begin{proof}
Put $N=16L^2$. On the one hand,
\[
 \binom{2N}{N}=\prod_{i=1}^{N}\frac{N+i}{i}\ge 2^N.
\]
On the other hand, the factorial valuation formula gives, for every prime
$\ell$,
\[
 v_\ell\binom{2N}{N}
 =\sum_{a\ge1}\left(
 \left\lfloor\frac{2N}{\ell^a}\right\rfloor
 -2\left\lfloor\frac{N}{\ell^a}\right\rfloor\right)
 \le\lfloor\log_\ell(2N)\rfloor.
\]
Each summand is either zero or one and vanishes when $\ell^a>2N$.
Thus each prime contributes a factor at most $2N$ to the binomial coefficient.
Writing $\pi(u)$ for the number of primes at most $u$, we obtain
\[
 2^N\le\binom{2N}{N}\le(2N)^{\pi(2N)},\qquad
 \pi(32L^2)\ge\frac{16L^2}{5+2\log_2L}\ge2L.
\]
For the last inequality, $\log_2L\le L-1$ implies
$5+2\log_2L\le2L+3\le8L$. Trial division of every integer through $32L^2$
uses polynomially many operations on integers with $O(\log(L+1))$ bits,
and hence polynomial total bit time in $L$.
\end{proof}

\begin{lemma}\label{lem:auxiliary}
Fix a prime $p$. Given $m\ge1$ and distinct primes
$q_1,\ldots,q_n\ne p$ of polynomial magnitude in $n+m$, one can
deterministically find a prime $q$, of polynomial magnitude in $n+m$, such that
\begin{equation}\label{eq:auxiliary}
 q\notin\{p,q_1,\ldots,q_n\},\qquad
 p^t\not\equiv1\pmod q\quad(1\le t\le m).
\end{equation}
In particular, $\ord_q(p)>m$.
\end{lemma}
\begin{proof}
Define the integers
\begin{equation}\label{eq:BL}
 c=\lceil\log_2p\rceil,\qquad
 B=\frac{cm(m+1)}2,\qquad L=B+n+2.
\end{equation}
Every prime failing one of the congruence tests divides
$R=\prod_{t=1}^m(p^t-1)$. The number of its distinct prime divisors is at most
\[
 \log_2R\le\sum_{t=1}^m t\log_2p\le B.
\]
If $R=1$, there are no such divisors, and the same bound holds. After also
excluding $p,q_1,\ldots,q_n$, at most $B+n+1=L-1$ primes are forbidden.
Lemma~\ref{lem:primes} supplies at least $2L$ primes at most $32L^2$.
Enumerate these primes and test the successive residues
$p,p^2,\ldots,p^m$ modulo each candidate by modular multiplication.
At least one candidate passes all tests. Since $p$ is fixed, $L$ and the
search range are polynomial in $n+m$, and enumeration and all modular tests
have deterministic polynomial bit cost. The product $R$ serves only to count
excluded primes; its factorization is not required.
\end{proof}

\subsection{Noncoprimality over Field of Positive Characteristic is NP-hard}
\label{sec-charp}
In this section, we show that noncoprimality over fields of positive characteristic is NP-hard.
Note that this is different from the results in \cite{Gathen1993CC,kaltofen2005complexity}, where random reductions are used.

We first describe how roots of unity encode assignments. The same encoding
will also be used over the complex numbers in the characteristic-zero case.

\begin{lemma}\label{lem:boolean}
Let $q_1,\ldots,q_n$ be distinct primes and $D=\prod_iq_i$. Work in an
algebraically closed field whose characteristic does not divide $D$. For a
root $y$ of $Y^D-1$, declare variable $i$ true if and only if
$y^{D/q_i}=1$. Every Boolean assignment is realized by such a root. Moreover,
\begin{equation}\label{eq:Si}
 S_i(Y)=\sum_{a=0}^{q_i-1}Y^{aD/q_i}
\end{equation}
takes the value $q_i$ when variable $i$ is true and $0$ when it is false.
\end{lemma}
\begin{proof}
The polynomial $Y^D-1$ has $D$ distinct roots, since its derivative does not
vanish at a root. These roots form a cyclic multiplicative group of order $D$.
For completeness, any finite multiplicative subgroup of a field is cyclic:
if its exponent is $e$, an element of order $e$ is obtained by combining
prime-primary parts of elements with maximal prime-power orders. Every group
element solves $Z^e-1=0$, so the group has at most $e$ elements and equals
the cyclic subgroup generated by that element.

Choose a generator $\zeta$ of the $D$th roots. For $y=\zeta^a$,
$y^{D/q_i}=1$ holds precisely when $a\equiv0\pmod {q_i}$. Prescribe the
residue $0$ for a true variable and $1$ for a false variable. The integer
Chinese remainder theorem produces a simultaneous choice of $a$, proving
that every assignment occurs. Finally, $u=y^{D/q_i}$ is a $q_i$th root.
The sum $1+u+\cdots+u^{q_i-1}$ is $q_i$ if $u=1$ and zero otherwise,
by the geometric-sum identity.
\end{proof}

Fix a field $\F$. A computable presentation supplies finite-string names of
field elements and computable field operations and equality. We make no
polynomial-time assumption on these procedures. 
Sparse polynomials are lists of nonzero coefficient names paired with distinct nonnegative exponents in increasing order; the exponents are written in binary. 
All lengths and running times below are bit lengths and total bit running times.
Outputs are explicit normalized sparse lists: a nonzero GCD is monic and
$\gcd(0,0)=0$. The polynomial $1$ has one fixed, unpadded output encoding.
Other coefficient names need not be globally unique.

\begin{proposition}\label{prop:positive}
For every fixed computable field $\F$ of characteristic $p>0$,
3SAT reduces deterministically in polynomial bit time to noncoprimality of
two sparse polynomials in $f,g\in\F[X]$. 
All coefficients use names from a fixed finite dictionary. 
%This includes $p=2$.
\end{proposition}
\begin{proof}
Use the notation $\Phi,n,m$  for 3SAT in Section \ref{sec-setup}.
Choose $n$ distinct primes $q_i\ne p$ at most $32(n+1)^2$, using
Lemma~\ref{lem:primes}. Choose $q$ by Lemma~\ref{lem:auxiliary} and set
\[
 D=\prod_{i=1}^nq_i.
\]
Then $p\nmid Dq$, $\gcd(D,q)=1$, and $\ord_q(p)>m$.
All polynomial coefficients will lie in the prime subfield $\Fp$ of $\F$.

\emph{Boolean values and clause violations.}
Since $q_i$ is nonzero modulo $p$, define
\begin{equation}\label{eq:Bi}
 B_i(Y)=q_i^{-1}S_i(Y)
       =q_i^{-1}\sum_{a=0}^{q_i-1}Y^{aD/q_i}\in\Fp[Y].
\end{equation}
By Lemma~\ref{lem:boolean}, at a $D$th root $y$ the value $B_i(y)$ is
exactly the Boolean value of variable $i$, with true represented by $1$ and
false by $0$.

Number the clauses $j=0,\ldots,m-1$. For each clause, define $V_j(Y)$ to be
the product of its literal-violation factors: a positive literal on variable
$i$ contributes $1-B_i(Y)$, and a negative literal contributes $B_i(Y)$.
Use $V_j=1$ for an empty clause. Each product has at most three factors, and
at every $D$th root,
\begin{equation}\label{eq:violations}
 V_j(y)\in\{0,1\}\subset\Fp,\qquad
 V_j(y)=1\quad\Longleftrightarrow\quad\text{clause $j$ is false}.
\end{equation}
Thus, the assignment encoded by $y$ satisfies $\Phi$ exactly when all the
$V_j(y)$ vanish.

\emph{Packing the clauses into a pair of polynomials.}
Construct
\begin{equation}\label{eq:positive-pair}
 f(X)=\sum_{a=0}^{q-1}X^{aD},\qquad
 g(X)=\sum_{j=0}^{m-1}X^{Dj}V_j(X^q).
\end{equation}
We prove the required equivalence using roots in an algebraic closure of $\F$.

First, suppose $f,g$ have a non-trivial GCD.
Then they have a a common root $\alpha$  and put $\beta=\alpha^D$ and
$y=\alpha^q$. The equation $f(\alpha)=0$ says
$1+\beta+\cdots+\beta^{q-1}=0$. Consequently,
\[
 \beta^q=1,\qquad \beta\ne1,
\]
since at $\beta=1$ this sum is the nonzero field element $q$. Thus $\beta$
has order $q$, and $y^D=\beta^q=1$. The other root equation becomes
\begin{equation}\label{eq:frobenius-polynomial}
 P(\beta)=0,\qquad
 P(Z)=\sum_{j=0}^{m-1}V_j(y)Z^j\in\Fp[Z].
\end{equation}
The prime-field membership here follows from the Boolean evaluations
\eqref{eq:violations}, even when $\F$ is larger than its prime subfield.

Raising a vanishing equation with coefficients in $\Fp$ to $p$th powers
shows that every $\beta^{p^t}$ is also a root of $P$. In particular, the $m+1$ roots
$\beta,\beta^p,\ldots,\beta^{p^m}$ are distinct. Indeed, equality of the terms indexed by $s<t$
would give $p^{t-s}\equiv1\pmod q$, because $\beta$ has order $q$ and
$p$ is invertible modulo $q$. This contradicts \eqref{eq:auxiliary}.
Since $P$ has degree at most $m-1$ unless it is zero, it must be the zero
polynomial. All coefficients $V_j(y)$ therefore vanish, and $\Phi$ is
satisfied.

Conversely, suppose $\Phi$ is satisfiable. Choose a $D$th root $y$ realizing
a satisfying assignment by Lemma~\ref{lem:boolean}. Choose a primitive
$q$th root $\beta$, which exists because $q\ne p$. Since $\gcd(D,q)=1$,
choose integers $a,b$ satisfying $aD+bq=1$ and set
\begin{equation}\label{eq:alpha}
 \alpha=\beta^a y^b.
\end{equation}
The roots are nonzero, so negative powers are defined. Using $\beta^q=1$
and $y^D=1$ gives
\[
 \alpha^D=\beta^{aD}y^{bD}=\beta,\qquad
 \alpha^q=\beta^{aq}y^{bq}=y.
\]
It follows that $f(\alpha)=0$ and
$g(\alpha)=\sum_j\beta^jV_j(y)=0$.

The polynomial $\F$ is nonzero. Over a field, such a pair has a common root
in the algebraic closure if and only if its gcd has positive degree: a
B\'ezout identity excludes a common root when the gcd is $1$, and a
positive-degree common divisor has a root in the algebraic closure.
This argument also covers $g=0$. Hence \eqref{eq:reduction} holds.

\emph{Sparse size and bit cost.}
Let $Q=\max(q,q_1,\ldots,q_n)$. Lemmas~\ref{lem:primes}
and~\ref{lem:auxiliary} bound $Q$ polynomially in $n+m$, with constants
allowed to depend on the fixed $p$. The integer $D$ has $O(n\log Q)$ bits.
Each $B_i$ has at most $Q$ terms, each $1-B_i$ at most $Q+1$, and each
$V_j$ at most $(Q+1)^3$ candidate monomial occurrences on expansion.
Therefore $g$ has at most $m(Q+1)^3$ candidate occurrences and $\F$ has $q$
terms. The respective exponent bounds are
\begin{equation}\label{eq:positive-size}
 \deg V_j\le3D,\qquad
 \deg g\le D(m-1+3q)\ \ (g\ne0),\qquad
 \deg f=D(q-1).
\end{equation}
Thus all exponents require only $O(n\log Q+\log(m+q))$ bits.

Expand the at-most-three-factor products, sort candidate terms by exponent,
merge equal exponents using residue arithmetic modulo $p$, and delete zero
coefficients. All lists have polynomial length and all integer operands have
polynomial bit length, so these steps take deterministic polynomial bit time.
Fix once and for all one finite-string name in $\F$ for each nonzero element
of $\Fp$. The final coefficients are copied from this dictionary. Its size
and the lengths of its names depend only on $\F$; all coefficient calculations
in the construction use residues modulo the fixed integer $p$. In particular,
no operation or equality test on arbitrary names in the presentation of $\F$
is performed. The output of the reduction consists of structurally parseable
sparse lists with sorted distinct binary exponents and nonzero coefficient
names, of polynomial total encoded length.
\end{proof}

The roots and the element \eqref{eq:alpha} are used to establish correctness,
not as objects that the reduction must compute or encode. The executable
construction consists only of prime searches, binary integer operations,
fixed-residue arithmetic, and copying coefficient names. None of the argument
excludes characteristic two: the distinct Boolean values $0,1$, the choice
$q_i,q\ne p$, and the Frobenius separation apply there as well.

\subsection{Noncoprimality over Field of Characteristic Zero is NP-hard}
\label{sec-charz}
In this section, we show that noncoprimality over fields of characteristic zero is NP-hard.
Note that this has been proved in \cite{plaisted1984}.
The following proof follows the main idea of that in \cite{plaisted1984}, but it is different in the construction of $H_i(X)$.
We present this new proof for self-consistency.

In characteristic zero we first work over $\Q$ and evaluate at complex roots
of unity. This permits the clause equations to be combined by a sum of
squared moduli. To retain fixed coefficient names, the resulting signed
monomial occurrences are kept separate rather than collected into potentially
large integer coefficients.

\begin{proposition}\label{prop:zero}
For every fixed computable field $\F$ of characteristic zero, 3SAT
reduces deterministically in polynomial bit time to noncoprimality of two
sparse polynomials $f,g\in\F[X]$ whose coefficients are only $1$ and $-1$.
\end{proposition}
\begin{proof}
Use the notation $\Phi,n,m$ for 3SAT in Section \ref{sec-setup}.
Choose distinct primes $q_1,\ldots,q_n\le32(n+1)^2$ using
Lemma~\ref{lem:primes}, and put
\[
 D=\prod_iq_i,\qquad Q=\max_iq_i.
\]
For a positive literal on variable $i$, use
$A_i(X)=X^{D/q_i}-1$; for a negative literal use the geometric sum $S_i(X)$
from \eqref{eq:Si}. On complex $D$th roots, $A_i$ vanishes exactly when the
variable is true and $S_i$ vanishes exactly when it is false. These are
therefore vanishing polynomials for the respective literal truth conditions.
For clause $j$, let $H_j(X)$ be the product of its at most three literal
polynomials, with the empty product equal to $1$. Then $H_j(y)=0$ if and
only if that clause is satisfied by the assignment encoded by $y$.

Consider the Laurent polynomial expression
\begin{equation}\label{eq:sum-squares}
 C(X)=\sum_{j=1}^m H_j(X)H_j(X^{-1}).
\end{equation}
Every complex $D$th root $y$ lies on the unit circle, and the $H_j$ have real
coefficients. Hence $H_j(y^{-1})=\overline{H_j(y)}$, giving
\begin{equation}\label{eq:positive-sum}
 C(y)=\sum_{j=1}^m|H_j(y)|^2.
\end{equation}
This value is zero exactly when all clauses are satisfied.
Lemma~\ref{lem:boolean} ensures that all assignments occur among these roots.

Expand \eqref{eq:sum-squares} as a list of signed monomial occurrences,
without collecting equal exponents:
\begin{equation}\label{eq:occurrences}
 C(X)=\sum_{t=0}^{T-1}\epsilon_tX^{e_t},\qquad
 \epsilon_t\in\{1,-1\},\qquad
 T\le m\max(2,Q)^6,\qquad |e_t|\le3D.
\end{equation}
Each literal polynomial has at most $\max(2,Q)$ terms. Each product
$H_j(X)H_j(X^{-1})$ contains at most six such factors, giving the occurrence
bound. The exponents in $H_j$ lie between $0$ and $3D$, so their differences
have absolute value at most $3D$, as asserted.

For each occurrence let $r_t\in\{0,\ldots,D-1\}$ be the residue of $e_t$
modulo $D$, and construct
\begin{equation}\label{eq:signed-pair}
 f(X)=X^D-1,\qquad
 g(X)=\sum_{t=0}^{T-1}\epsilon_tX^{r_t+tD}.
\end{equation}
The exponent of the $t$th term lies in $[tD,(t+1)D)$, so these exponents
are distinct and already increasing. At every $D$th root $y$,
\[
 y^{r_t+tD}=y^{e_t},\qquad g(y)=C(y).
\]
Thus $f,g$ have a common complex root precisely when $\Phi$ is satisfiable.
The rational gcd has positive degree in exactly this case. Indeed, a rational
B\'ezout identity rules out common complex roots when the rational gcd is
$1$, and a positive-degree rational gcd has a complex root.

This equivalence holds for the same pair regarded in $\F[X]$. The prime
subfield of a characteristic-zero field is a copy of $\Q$. If the rational
gcd is $1$, its rational B\'ezout identity remains valid in $\F[X]$. If a
positive-degree rational polynomial divides both inputs, the same polynomial
is a positive-degree common divisor over $\F$. This proves
\eqref{eq:reduction} for every such $\F$.

Finally, fix names for $1$ and $-1$ in $\F$. Every coefficient in
\eqref{eq:signed-pair} is represented by a copy of one of these names.
The reduction never forms a name for a collected integer coefficient.
The signed occurrence list has polynomial length, and its residues and
shifted exponents are computed with binary integers. Specifically,
\begin{equation}\label{eq:zero-size}
 \log_2D\le n\log_2Q,\qquad
 0\le r_t+tD<TD,\qquad
 \log_2(TD)=O(\log(T+1)+n\log Q).
\end{equation}
Trial division, occurrence expansion, reduction modulo $D$, and the exponent
shifts all have polynomial total bit cost. The two final sparse lists have
polynomial encoded length, distinct nonnegative binary exponents, and
nonzero coefficients. No computation in the presentation of $\F$ is needed.
\end{proof}

\subsection{Output-sensitive Sparse GCD is NP-hard}
\label{sec-proof}
IN this section, we will prove  Theorem~\ref{thm-m1}, that is, output-sensitive sparse GCD is NP-hard.

The coefficient restrictions in the preceding constructions are part of their bit-complexity arguments. In characteristic $p$, only $p-1$ fixed names are needed; in characteristic zero, only the names of $1$ and $-1$ are required.
These dictionaries can be built into the reduction for the fixed field.
The cost of representing arbitrary field elements, therefore, never enters the construction. Nonuniqueness of names does not matter because the reduction always copies the selected names.

The scalar-extension argument also has a limited role. In characteristic zero
it transfers coprimality, or a positive-degree common divisor, from $\Q$ to
$\F$ without carrying out field computations. In positive characteristic the
proof works in an algebraic closure of $\F$, and the evaluated clause values
lie in $\Fp$ regardless of the size of $\F$. Consequently the Frobenius
argument does not require the auxiliary root to have large degree over $\F$;
it only uses a vanishing polynomial with coefficients in $\Fp$.

It remains to explain why allowing a running time polynomial in output
length does not avoid the obstruction. There is no upper bound here on the
length of the gcd in satisfiable instances. Only the fixed output length of
$1$ is used.

\begin{lemma}\label{lem:cutoff}
Fix a field and the stated encoding conventions. Suppose 3SAT admits a
deterministic polynomial-bit-time reduction to pairs $(f,g)$ with $f\ne0$ satisfying
\eqref{eq:reduction}. If a deterministic exact GCD algorithm has running time
polynomial in the sum of input and output lengths, then $\PP=\NP$.
\end{lemma}
\begin{proof}
Let $A$ be the proposed algorithm, and let $c_F$ be the length of the single
required output encoding of $1$. There are integers $C_0,d\ge1$ such that
on every input of encoded length $N$ and output length $M$, the running time
of $A$ is at most
\begin{equation}\label{eq:timebound}
 C_0(1+N+M)^d.
\end{equation}
These constants may depend on $\F$ and $A$ and may be built into a decision
algorithm. No uniform procedure for finding them from a field presentation
is needed.

Given $\Phi$, construct $f,g$, determine their encoded input length $N$,
and simulate $A$ for at most
\begin{equation}\label{eq:cutoff}
 \tau=C_0(1+N+c_F)^d
\end{equation}
steps. If $A$ halts within this limit and returns exactly the fixed encoding
of $1$, declare $\Phi$ unsatisfiable. In every other case, including failure
to halt within the limit, declare $\Phi$ satisfiable.

If $\Phi$ is unsatisfiable, the normalized gcd is $1$. Its output has length
exactly $c_F$, so \eqref{eq:timebound} forces $A$ to halt by the cutoff with
exactly the prescribed string. If $\Phi$ is satisfiable, its gcd has positive
degree. Exactness prevents $A$ from returning that string, whether it halts
before the cutoff or runs longer. Both answers are therefore correct.

The simulation and output comparison take time polynomial in $\tau$, and
$N$ is polynomial in the formula length. Thus this is a deterministic
polynomial-time algorithm for 3SAT. By its NP-completeness, $\PP=\NP$.
\end{proof}

The comparison in this proof is a comparison of strings, not a field
equality test. It uses the specified encoding of $1$ and does not require
canonical names for all coefficients. It also does not compute, bound, or
inspect the full output of $A$ on runs that exceed the cutoff.

%\section{Assembly over every fixed field}

Finally, we prove Theorem~\ref{thm-m1} by assembling over every fixed field.

\begin{proof}[Proof of Theorem~\ref{thm-m1}]
Fix any computably presented field $\F$ satisfying the stated representation
conditions. Its characteristic is either zero or a fixed prime $p$.
In characteristic zero, Proposition~\ref{prop:zero} supplies a deterministic
polynomial-bit-time reduction satisfying \eqref{eq:reduction}. In positive
characteristic, Proposition~\ref{prop:positive} supplies such a reduction,
including when $p=2$. The directly decidable degenerate formulas are covered
by the fixed pairs stated in the introduction.

In both branches the first polynomial is nonzero, the inputs are valid
explicit sparse lists, and all coefficients have names from fixed finite
dictionaries. Suppose a deterministic exact gcd algorithm with the claimed
output-sensitive bound existed for this $\F$. Lemma~\ref{lem:cutoff} would
then put 3SAT in $\PP$, contradicting $\PP\ne\NP$.

The choice of reduction, coefficient names, algorithm, and polynomial
constants has been allowed to depend throughout on the fixed $\F$. No
polynomial-time assumption on field operations, equality, or the production
of names for arbitrary elements has been used. All first polynomials in the
reduction are nonzero, so the convention $\gcd(0,0)=0$ introduces no further
case. The contradiction applies to every fixed field in the theorem.
\end{proof}

%  Section 2.6 --- A strengthening of Lemma~\ref{lem:cutoff}
%  To be inserted after the proof of Theorem~\ref{thm-m1}
%  (Section~\ref{sec-proof}) and before Section~\ref{sec-ur}.
%  Deliberately kept as a single unnumbered block (no subsubsections):
%  the strengthened obstruction is the main content, the product regime
%  and the parameterized remark are only mentioned.
%
%  External labels used (all defined in the paper body):
%     sec-setup      -> Section "Proof Sketch"
%     sec-pre        -> Section "Preliminary Results of Primes"
%     sec-proof      -> Section "Proof of Theorem~\ref{thm-m1}"
%     lem:cutoff     -> Lemma "cutoff" (output-sensitive <=> P = NP)
%     prop:positive  -> Proposition, positive characteristic
%     prop:zero      -> Proposition, characteristic zero
%     eq:timebound   -> the bound C_0(1+N+M)^d
%     eq:reduction   -> Phi satisfiable <=> deg gcd > 0
%     thm-m1         -> Theorem "main" (output-sensitive sparse GCD is NP-hard)
%     thm-RUD-F      -> Roots-of-Unity-Detection theorem
%  Notation follows the paper's preamble: \F is the field macro,
%  \PP and \NP are \mathsf{P} and \mathsf{NP}, \norm{\cdot} is the term count.
% =====================================================================

\subsection{Sparse polynomial GCD is not Fixed-output Polynomial Time}
\label{sec-strengthening}

Lemma~\ref{lem:cutoff} converts an output-sensitive algorithm into a polynomial-time decision procedure for \textsc{3SAT} by exploiting a single fact: on \emph{unsatisfiable} instances, the GCD is exactly~$1$, whose encoding has a \emph{fixed constant} length of $c_F$.  
In this section, we show that this observation leads to  stronger complexity results using the language of parameterized complexity.

%We show here that this observation is far more robust than the specific bound $C_0(1+N+M)^d$ of~\eqref{eq:timebound} suggests: the obstruction survives for any running-time regime in which the output length is held fixed, even if the dependence on $M$ is arbitrarily violent. 
%Consequently, the hardness does not lie in ``the output being large'', but in ``deciding the constant-output case''.  
Recall that $N$ denotes the encoded input length and $M$ the encoded output length of an instance $(f,g)$.
Denote $T(N, M)$ as the complexity of the output-sensitive GCD algorithm. 

\begin{definition}%[Fixed-output polynomial regime]
\label{def:fpo}
A deterministic exact \textsc{GCD} algorithm $A$ runs in 
\emph{fixed-output polynomial} time if, for every fixed value
$M_0 \in \N$, the function
\[
N \;\longmapsto\; T(N, M_0)
\]
is bounded by a polynomial in $N$; that is,
\[
\forall\, M_0, \;\exists\, p_{M_0} \in \N[x], \;\;\forall\, N:\qquad
T(N, M_0) \;\le\; p_{M_0}(N).
\]
The degree of $p_{M_0}$ is allowed to depend on $M_0$ without restriction.
\end{definition}

%\begingroup\emergencystretch=2em\relax
The regime of Definition~\ref{def:fpo} is strictly more permissive than that
of Lemma~\ref{lem:cutoff}.  For instance, 
%$T(N,M) = N^M$, $T(N,M) = N^{2^M}$ and 
$T(N,M) = N^{2^M}$ satisfies Definition~\ref{def:fpo}, whereas it is not bounded by $C_0(1+N+M)^d$.
%\endgroup
We have the following strengthening of Lemma~\ref{lem:cutoff}.
\begin{theorem}
\label{thm-strengthened}
Assume $\PP\ne\NP$. Fix a computable field $\F$ and the encoding
conventions of Section~\ref{sec-setup}.  There is no deterministic exact algorithm that computes $\gcd(f,g)$ for all sparse $f,g \in \F[X]$ and runs, on every instance, within a time bound $T(N,M)$ satisfying the fixed-output polynomial regime of Definition~\ref{def:fpo}.
\end{theorem}

\begin{proof}
Fix any such algorithm $A$.  Apply the deterministic polynomial-bit-time
reduction of Proposition~\ref{prop:positive} (characteristic $p>0$) or
Proposition~\ref{prop:zero} (characteristic zero) to a \textsc{3SAT}
instance $\Phi$ of encoded length $s$, obtaining sparse polynomials
$f, g \in \F[X]$ with $f \ne 0$ satisfying~\eqref{eq:reduction}.  By those propositions the reduction runs in polynomial bit time and the encoded input
length $N = \Size(f,g)$ satisfies
\begin{equation}
N \;\le\; K_0\, s^{\kappa}
\label{eq:N-poly}
\end{equation}
for constants $K_0, \kappa \ge 1$ depending only on the fixed field $\F$.

Let $c_F$ denote, as in Lemma~\ref{lem:cutoff}, the length of the single required output encoding of $1$.  
%The value $c_F$ is a \emph{constant}: it is the length of one fixed, unpadded string and does not depend on $s$, $N$ or $\Phi$.  
By Definition~\ref{def:fpo} applied to the constant $M_0 := c_F$,
there is a polynomial $p_{c_F}$ with
\[
T(N, c_F) \;\le\; p_{c_F}(N) \qquad \text{for all } N .
\]

Now decide $\Phi$ as follows.  Compute the pair $(f,g)$ via the reduction,
let $N$ be its encoded length, and simulate $A$ for at most
\begin{equation}
\tau \;=\; p_{c_F}(N)
\label{eq:new-cutoff}
\end{equation}
steps.  Declare $\Phi$ \emph{unsatisfiable} if $A$ halts within this limit and outputs exactly the fixed encoding of $1$; declare $\Phi$ \emph{satisfiable} in every other case, including non-halting within the limit.
The correctness argument is identical to that of Lemma~\ref{lem:cutoff}.
%Write $\delta := \deg\gcd(f,g)$.  If $\Phi$ is unsatisfiable, then $\delta = 0$ by~\eqref{eq:reduction}, and with the stated normalization $\gcd(f,g) = 1$; its output length is exactly $c_F$, so the bound $T(N,c_F) \le p_{c_F}(N) = \tau$ forces $A$ to halt within the cutoff and to return the prescribed string.  If $\Phi$ is satisfiable, then $\delta > 0$, so exactness prevents $A$ from returning the encoding of $1$---whether it halts before $\tau$ or not.  Both answers are therefore correct.
%%
%Running the reduction, computing $\tau$, simulating at most $\tau$ steps and comparing the output string all cost time polynomial in $\tau$ and in $s$. By~\eqref{eq:new-cutoff} and~\eqref{eq:N-poly}, $\tau$ is a polynomial in $N$, hence in $s$.  Thus \textsc{3SAT} is decided in deterministic polynomial time, so $\PP=\NP$, contrary to the hypothesis.
\end{proof}

%\begin{remark}[Why only the value $M_0 = c_F$ is consumed]
%\label{rem:only-one-M}
%The proof uses the hypothesis of Definition~\ref{def:fpo} at exactly one constant, $M_0 = c_F$.  No uniform polynomial across all $M_0$, and no monotonicity of $M \mapsto T(N,M)$, is required.  This is possible precisely because the \emph{no}-instances of the reduction are \emph{constant-output} instances: the reduction never has to evaluate $A$ on an instance whose output is large in order to certify unsatisfiability.  A function such as $T(N,M) = N^{(M-1)\,2^{N}} + N$ exhibits $T(N,1)=2N$, yet violates Definition~\ref{def:fpo} and is consistent with Theorem~\ref{thm-strengthened}; conversely $T(N,M)=2^{NM}$ violates Definition~\ref{def:fpo} already at $M_0=1$ and is excluded by Lemma~\ref{lem:cutoff} as well.  The regime of Definition~\ref{def:fpo} is exactly the regime in which the constant-output test remains cheap.
%\end{remark}

\medskip
The obstruction is not tied to the shape of the bound.  A separable special case deserves explicit mention.
%, being the form in which output-sensitive complexity is usually quoted: suppose the running time obeys
In the language of parameterized complexity, the GCD problem is in the {\em class $\mathrm{FPT}_{M}$} (fixed-parameter tractable), if  
\begin{equation}
T(N,M) \;\le\; a(M)\cdot b(N),
\label{eq:product}
\end{equation}
with $b \in \N[x]$ a polynomial and $a : \N \to \N$ an arbitrary computable function.
It is clear that $\mathrm{FPT}_{M}$ satisfies Definition~\ref{def:fpo}. We thus have
%In the language of parameterized complexity,  the problem is in the class $\mathrm{FPT}_{\nu}$, where
%\[\nu(f,g) \;:=\; \norm{\gcd(f,g)}_0\]
%is the \emph{number of terms of the gcd}---the natural output-sensitive parameter, polynomially related to $M$ whenever exponents are represented in binary and coefficients come from a fixed dictionary.

\begin{corollary}
\label{cor:not-fpt}
Assume $\PP\ne\NP$. Sparse output-sensitive \textsc{GCD} over any fixed computable field $\F$ is \emph{not} in $\mathrm{FPT}_{M}$: there is no algorithm solving it in time $a(M)\cdot \poly(N)$ for any computable $a$ abd polybinomial $b$.
%, and a fortiori none in time $2^{o(\nu)} \cdot N^{O(1)}$.
\end{corollary}

%\begin{proposition}
%\label{prop:product}
%Assume $\PP\ne\NP$.  Then no deterministic exact \textsc{GCD} algorithm for sparse $f,g \in \F[X]$ obeys~\eqref{eq:product} for any polynomial $b \in \N[x]$ and any nondecreasing unbounded $a$.
%\end{proposition}
%
%\begin{proof}
%Set $M_0 := c_F$.  Since $c_F$ is a constant, $a(c_F)$ is a finite positive constant, and~\eqref{eq:product} gives
%\[
%T(N, M_0) \;\le\; b(N)\cdot a(c_F) \;=\; O\bigl(b(N)\bigr),
%\]
%which is a polynomial in $N$.  Hence at $M_0 = c_F$ the bound already implies Definition~\ref{def:fpo}, and Theorem~\ref{thm-strengthened} applies.
%\end{proof}

Requiring $b$ to be a polynomial is essential: the argument evaluates the bound at a fixed output length, so it is the \emph{input-side} factor that must be polynomial.  
Were $b$ allowed to grow super-polynomially, say $b(N)=2^N$ with $a(M)=1$, then $T(N,M)=2^N$ would satisfy~\eqref{eq:product}
while remaining consistent with $\PP\ne\NP$, and no contradiction would follow.  The regime~\eqref{eq:product} is nevertheless strictly more permissive than $C_0(1+N+M)^d$. 
For example, $T(N,M) = N^{10}\cdot 2^{M}$ satisfies~\eqref{eq:product} with $b(N) = N^{10}$ and $a(M) = 2^{M}$, but is not bounded by any polynomial in $N+M$; it is nevertheless excluded because, at the constant output length $c_F$, the factor $2^{M}$ collapses to the constant $2^{c_F}$, and only the polynomial factor $N^{10}$ remains. 
By contrast, $T(N,M) = N^{M}$ is not of the form~\eqref{eq:product} at all, since
a single polynomial $b$ cannot dominate $N^M$ uniformly in $M$; it is covered
by Definition~\ref{def:fpo} instead.  Thus
\[
\underbrace{C_0(1+N+M)^d}_{\text{Lemma~\ref{lem:cutoff}}}
\;\subsetneq\;
\underbrace{b(N)\,a(M),\ \eqref{eq:product}}_{\text{Corollary
\ref{cor:not-fpt}}}\;\subsetneq\;
\underbrace{\text{Definition~\ref{def:fpo}}}_{\text{Theorem~\ref{thm-strengthened}}},
\]
each inclusion being strict.
Thus, Theorem~\ref{thm-strengthened} and Corollary \ref{cor:not-fpt} are stronger than the NP-hardness result.

\begin{remark}
From the proofs, it is clear that $M$ could be replaced by the number of terms of the GCD in Theorem~\ref{thm-strengthened} and Corollary \ref{cor:not-fpt}.
\end{remark}

\section{Roots of Unity Detection Is NP-hard over Finite Field}
\label{sec-ur}
By \textit{Roots of Unity Detection}, we refer to the following problem: given a sparse polynomial \(f \in \mathbb{F}_q[x]\) of degree \(D\), determine whether \(f(x)\) has, as a root, a root of unity whose order is at most \(D\).
%In this section, we show that Roots of Unity Detection Is NP-hard.
In this section, we present the proof for Theorem \ref{thm-m2}, that is, Roots of Unity Detection over finite field is NP-hard. 
%The proof consists of four steps given in the four subsections, respectively.
We first present the main result and then the lemmas.

%\subsection{Roots of Unity Detection Is NP-hard over Finite Field}
%In this section, we will prove the following theorem.
\begin{theorem}
\label{thm-RUD-F}
The following decision problem is hard in the randomized-reduction/BPP-consequence sense:
given a sparse polynomial \(f\in \F_q[x]\), an integer \(n\), the polynomial
$g=x^n-1,$
and a degree bound \(D\) with
\[
  \deg f\le D,\qquad n\le D,
\]
decide whether \(f\) and \(g\) are coprime. Equivalently, decide whether
$  \deg \gcd(f,g)>0.$
More precisely, a polynomial-time algorithm for this problem would
imply \(\mathrm{SAT}\in\mathrm{BPP}\), and therefore
$\mathrm{NP}\subseteq\mathrm{BPP}$.
\end{theorem}
\begin{proof}
The reduction is given by Lemmas~\ref{lem:source}--\ref{lem:amplification}.
For each randomized trial, Lemma~\ref{lem:source} produces a bounded-clause
formula whose unsatisfiability is zero-preserved and whose satisfiable inputs
are isolated with inverse-polynomial probability. Lemmas~\ref{lem:field-domain}
and~\ref{lem:encoding} construct, over a prime field, a polynomial system whose
common roots over the root set of \(x^E-1\) are precisely the satisfying
assignments of the isolated formula. Lemma~\ref{lem:compression-GCD} replaces
that system by one sparse polynomial \(H\) with inverse-polynomially small
failure probability and identifies positive GCD degree with the presence of a
satisfying assignment on the successful compression event.
Lemma~\ref{lem:instance-bounds} checks that the constructed instance has the
required restricted form \(g=x^n-1\), with \(n=E\) and \(D=2E\), and is
represented in polynomial sparse bit size. Finally,
Lemma~\ref{lem:amplification} repeats the trial polynomially many times and
shows that a polynomial-time solver for the restricted problem gives a BPP
algorithm for SAT. Since every language in \(\mathrm{NP}\) has a polynomial-time
many-one reduction to SAT, this implies \(\mathrm{NP}\subseteq\mathrm{BPP}\).
\end{proof}
\begin{remark}[Scope]
\label{rem:scope}
The theorem is stated in the exact scope used by the reduction: variable prime
fields, randomized reductions with a BPP consequence, target polynomial
\(g=x^n-1\), and target degree bound \(D=2n\). It does not assert deterministic
Karp hardness, fixed-field hardness, or the stronger side condition
\(\deg f\le n\).
\end{remark}

We now present the necessary lemmas along with their proofs.
Observe that the proof may be arranged in the following sequence:
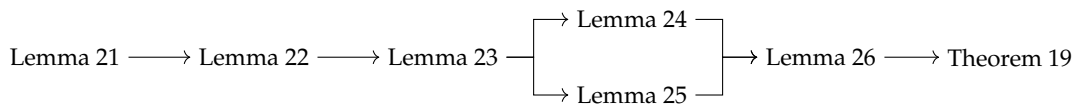
\begin{figure}[h]
    \centering
    \small
    \begin{tikzpicture}
        \node (lem15) at (-6,0) {Lemma~\ref{lem:source}};
        \node (lem16) at (-3.5,0) {Lemma~\ref{lem:field-domain}};
        \node (lem17) at (-1,0) {Lemma~\ref{lem:encoding}};
        \node (lem18) at (1.5,0.5) {Lemma~\ref{lem:compression-GCD}};
        \node (lem19) at (1.5,-0.5) {Lemma~\ref{lem:instance-bounds}};
        \node (lem20) at (4,0) {Lemma~\ref{lem:amplification}};
        \node (thm13) at (6.5,0) {Theorem~\ref{thm-RUD-F}};

        \draw [->] (lem15) -- (lem16);
        \draw [->] (lem16) -- (lem17);
        \draw [->] (lem17) --++(1.2,0) |- (lem18);
        \draw [->] (lem17) --++(1.2,0) |- (lem19);
        \draw [->] (lem18) --++(1.2,0) |- (lem20);
        \draw [->] (lem19) --++(1.2,0) |- (lem20);
        \draw [->] (lem20) -- (thm13);
    \end{tikzpicture}
    \label{fig:pf-sequence}
\end{figure}
% \begin{align*}
% &\text{Lemma~\ref{lem:source}}
% \longrightarrow
% \text{Lemma~\ref{lem:field-domain}}
% \longrightarrow
% \text{Lemma~\ref{lem:encoding}},\\
% &\text{Lemma~\ref{lem:encoding}}
% \longrightarrow
% \text{Lemma~\ref{lem:compression-GCD}},\qquad
% \text{Lemma~\ref{lem:encoding}}
% \longrightarrow
% \text{Lemma~\ref{lem:instance-bounds}},\\
% &\text{Lemmas~\ref{lem:compression-GCD} and~\ref{lem:instance-bounds}}
% \longrightarrow
% \text{Lemma~\ref{lem:amplification}}
% \longrightarrow
% \text{Theorem~\ref{thm-RUD-F}}.
% \end{align*}

We use the notion of randomized polynomial-time reduction as defined by Valiant and Vazirani in their seminal paper \citep{VV1986}. This notion allows for a one-sided error: no-instances are always mapped to no-instances, while yes-instances are mapped to yes-instances with at least inverse-polynomial probability.

\begin{definition}[Randomized polynomial-time reduction \citep{VV1986}]\label{def-6}
A problem \(A\) is \emph{randomized polynomial-time reducible} to a problem \(B\) if there exists a randomized Turing machine \(T\) running in polynomial time and a polynomial \(P\) such that:
\begin{enumerate}
    \item \textbf{(No-instance correctness):} If \(x \notin A\), then for all coPin flips, \(T(x) \notin B\).
    \item \textbf{(Yes-instance correctness):} If \(x \in A\), then \(\Pr[T(x) \in B] \ge 1/P(|x|)\).
\end{enumerate}
Here \(|x|\) denotes the bit-length of the encoding of the input instance \(x\), i.e., the input size.
\end{definition}

Based on Theorem 1.1 of Valiant and Vazirani \citep{VV1986} and their notion of randomized polynomial-time reduction (see Definition \ref{def-6}), we obtain the following concrete isolation result, which serves as the core lemma in our reduction.

\begin{lemma}[Valiant--Vazirani isolation theorem \cite{VV1986}] \label{lm-11}
For a CNF formula \(F\) on \(n \ge 1\) Boolean variables, the randomized Valiant--Vazirani procedure adds random homogeneous linear equations over \(\mathbb{F}_2\) (equivalently, affine constraints after translating the selected cell), encoded by polynomial-size CNF gadgets. The satisfying assignments of the output are exactly those of \(F\) that lie in the selected affine subspace. Thus an unsatisfiable input remains unsatisfiable for every random choice. If \(F\) is satisfiable, the output has exactly one satisfying assignment with probability at least \(1/P(|F|)\), where \(P\) is a polynomial (depending only on the reduction) and \(|F|\) denotes the encoding length of \(F\).
\end{lemma}

\begin{lemma}[Isolation and parsimonious bounded-clause source]
\label{lem:source}
There is a randomized polynomial-time procedure which, given a SAT instance
\(\Phi\) of bit length \(s\), outputs a CNF formula
\[
  \Psi(y_1,\ldots,y_m)=C_1\wedge\cdots\wedge C_k
\]
with clause size at most three, where \(m,k=\poly(s)\), such that:
\begin{enumerate}
  \item if \(\Phi\) is unsatisfiable, then every possible output \(\Psi\) is
  unsatisfiable;
  \item if \(\Phi\) is satisfiable, then with probability at least \(1/A(s)\),
  for a polynomial \(A\), the output \(\Psi\) has exactly one satisfying
  assignment;
  \item the bounded-clause conversion is parsimonious, so it preserves the
  exact number of satisfying assignments.
\end{enumerate}
The isolation input is the zero-preserving Valiant--Vazirani isolation theorem
together with polynomial-size CNF encodings of affine parity constraints over
\(\F_2\), followed by the standard parsimonious conversion to clauses of size
at most three. 
%Bibliographic data for this source input are TODO.
\begin{proof}
Apply the zero-preserving Valiant--Vazirani isolation procedure to the SAT
instance. Its zero-preserving property gives the first item: if the original
formula has no satisfying assignment, then no output produced by the procedure
has a satisfying assignment. If the original formula is satisfiable, the
isolation guarantee gives, with probability at least \(1/A(s)\) for a
polynomial \(A\), a formula with exactly one satisfying assignment, after the
affine parity constraints have been encoded by polynomial-size CNF formulas
over Boolean variables.

It remains to keep the bounded-clause conversion parsimonious. Consider a long
clause
\[
  \lambda_1\vee\cdots\vee\lambda_t .
\]
Introduce auxiliary variables \(p_2,\ldots,p_{t-1}\), impose
\[
  p_2=\lambda_1\vee\lambda_2,\qquad
  p_j=p_{j-1}\vee\lambda_j\quad (3\le j\le t-1),
\]
encode each displayed OR equivalence by its three-clause CNF encoding, and add
the final clause \(p_{t-1}\vee\lambda_t\). For fixed values of the old
variables, the auxiliary variables are uniquely forced by the recurrence above.
The final clause is satisfied exactly when the original long clause is
satisfied. Thus every satisfying assignment of the old formula extends to
exactly one satisfying assignment of the converted formula, and no
non-satisfying assignment extends to a satisfying one. Clauses of size one or
two may be retained, or padded by repeated literals, without changing the
truth condition. Hence the conversion preserves the exact number of satisfying
assignments and yields clauses of size at most three.
\end{proof}
\end{lemma}

\begin{lemma}[Prime-field package and Boolean root domain]
\label{lem:field-domain}
Let \(\Psi(y_1,\ldots,y_m)\) be as in Lemma~\ref{lem:source}, and let
\(\varepsilon>0\) be inverse-polynomially small. Choose \(r_{\rm dum}\) with
\[
  2^{r_{\rm dum}}>1/\varepsilon,\qquad
  N=\max(m+r_{\rm dum},4).
\]
There is a Las Vegas randomized polynomial-time construction which outputs
distinct primes
\[
  \ell_1,\ldots,\ell_N,
\]
an integer \(c\ge 11\), and a prime
\[
  p=1+cM,\qquad M=\prod_{j=1}^N \ell_j,
\]
with \(\log p=\poly(N)\). Define
\[
  E=\prod_{i=1}^m \ell_i,\qquad
  P_{\rm dum}=\prod_{j=m+1}^N \ell_j.
\]
Then
\[
  E\mid p-1,\qquad
  \frac{E}{p}<\varepsilon,\qquad
  \log E\le \log p=\poly(s).
\]
Moreover one can construct, in Las Vegas expected polynomial time, an element
\(\omega\in\F_p^*\) of exact order \(E\). For
\[
  z_i=\omega^{E/\ell_i}\qquad (1\le i\le m),
\]
the element \(z_i\) has exact order \(\ell_i\). 
%The bibliographic data for the prime-field construction are TODO.
\begin{proof}
Apply the prime-field construction to the dummy \(N\)-variable formula
\[
  \bigwedge_{j=1}^N (u_j\vee u_j\vee u_j).
\]
It returns distinct primes \(\ell_1,\ldots,\ell_N\), an integer \(c\ge 11\),
and a prime \(p=1+cM\), where \(M=\prod_{j=1}^N\ell_j\), with
\(\log p=\poly(N)\). Since
\[
  p-1=cM=cEP_{\rm dum},
\]
we have \(E\mid p-1\). Also
\[
  \frac{E}{p}<\frac{E}{p-1}
  =\frac{1}{cP_{\rm dum}}\le \frac{1}{P_{\rm dum}}.
\]
The dummy primes are distinct and at least \(2\), so
\[
  P_{\rm dum}=\prod_{j=m+1}^N \ell_j\ge 2^{N-m}\ge 2^{r_{\rm dum}}>1/\varepsilon.
\]
Therefore \(E/p<\varepsilon\). The inequality \(\log E\le \log p\) follows
from \(E<p\), and \(\log p=\poly(N)=\poly(s)\) because
\(N=m+r_{\rm dum}\) up to the harmless lower bound \(4\), with
\(m=\poly(s)\) and \(r_{\rm dum}\) chosen inverse-polynomially in
\(\varepsilon\).

Since \(p\) is prime, \(\F_p^*\) is cyclic. Sample \(a\in\F_p^*\), put
\[
  \omega=a^{(p-1)/E},
\]
and test
\[
  \omega^{E/\ell_i}\ne 1\qquad (1\le i\le m).
\]
Because \(E=\prod_{i=1}^m\ell_i\) is squarefree, these tests hold exactly when
\(\omega\) has order \(E\). The accepted prime-field package gives
inverse-polynomial success probability for this sampling step, so repetition
gives a Las Vegas expected polynomial-time construction. Finally, if
\(\omega\) has order \(E\), then
\[
  z_i=\omega^{E/\ell_i}
\]
has order \(\ell_i\), since raising an element of order \(E\) to the
\(E/\ell_i\)-th power gives order \(E/\gcd(E,E/\ell_i)=\ell_i\).
\end{proof}
\end{lemma}

\begin{lemma}[Polynomial encoding of assignments and clauses]
\label{lem:encoding}
Assume the data of Lemma~\ref{lem:field-domain}. Put
\[
  R(x)=x^E-1.
\]
For \(1\le i\le m\), define
\[
  B_i(x)=\bigl(x^{E/\ell_i}-1\bigr)
         \bigl(x^{E/\ell_i}-z_i\bigr),
\]
and define literal polynomials
\[
  L_{i,1}(x)=x^{E/\ell_i}-z_i,\qquad
  L_{i,0}(x)=x^{E/\ell_i}-1,
\]
where \(L_{i,1}\) corresponds to \(y_i\) and \(L_{i,0}\) corresponds to
\(\neg y_i\). If
\[
  C=\lambda_1\vee\lambda_2\vee\lambda_3
\]
is a clause, with repeated literals allowed, define
\[
  P_C(x)=L_{\lambda_1}(x)L_{\lambda_2}(x)L_{\lambda_3}(x).
\]
Let \(F_1,\ldots,F_t\) be the list consisting of all \(B_i\) and all clause
polynomials \(P_C\) for clauses of \(\Psi\). Then
\[
  Z(R,F_1,\ldots,F_t)\subseteq \F_p
\]
is in bijection with the satisfying assignments of \(\Psi\).

More explicitly, the roots of \(R\) are \(\omega^a\), \(a\in\Z/E\Z\), and the
common roots of \(R,B_1,\ldots,B_m\) correspond exactly to residues satisfying
\[
  a\bmod \ell_i\in\{0,1\}\qquad (1\le i\le m).
\]
Under the convention
\[
  y_i=\mathrm{true}\Longleftrightarrow a\equiv 1\pmod{\ell_i},
  \qquad
  y_i=\mathrm{false}\Longleftrightarrow a\equiv 0\pmod{\ell_i},
\]
the polynomial \(P_C\) vanishes on such a root if and only if the corresponding
assignment satisfies the clause \(C\).
\begin{proof}
Since \(\omega\) has exact order \(E\), the roots of \(R=x^E-1\) in
\(\F_p\) are precisely
\[
  \omega^a,\qquad a\in \Z/E\Z .
\]
For \(x=\omega^a\), one has
\[
  x^{E/\ell_i}=(\omega^{E/\ell_i})^a=z_i^a.
\]
The element \(z_i\) has order \(\ell_i\), so
\[
  B_i(\omega^a)=0
  \Longleftrightarrow
  (z_i^a-1)(z_i^a-z_i)=0
  \Longleftrightarrow
  a\bmod \ell_i\in\{0,1\}.
\]
The primes \(\ell_i\) are distinct, and therefore the Chinese remainder
isomorphism
\[
  \Z/E\Z\simeq \prod_{i=1}^m \Z/\ell_i\Z
\]
identifies the common roots of \(R,B_1,\ldots,B_m\) with choices of residues
\(a\bmod \ell_i\in\{0,1\}\). Under the stated convention, these choices are
exactly Boolean assignments to \(y_1,\ldots,y_m\).

On a Boolean-domain root, \(L_{i,1}\) vanishes exactly when
\(a\equiv 1\pmod{\ell_i}\), namely when \(y_i\) is true, while \(L_{i,0}\)
vanishes exactly when \(a\equiv 0\pmod{\ell_i}\), namely when \(\neg y_i\) is
true. For a clause \(C=\lambda_1\vee\lambda_2\vee\lambda_3\), the polynomial
\[
  P_C=L_{\lambda_1}L_{\lambda_2}L_{\lambda_3}
\]
vanishes over the field \(\F_p\) exactly when at least one factor vanishes,
because \(\F_p\) is an integral domain. This is exactly the condition that at
least one literal of \(C\) is true. Thus imposing all the \(P_C\)'s cuts out
precisely the satisfying assignments of \(\Psi\) inside the Boolean root
domain, proving the asserted bijection.
\end{proof}
\end{lemma}

\begin{lemma}[Random compression and GCD correctness]
\label{lem:compression-GCD}
Choose independent uniform coefficients
\[
  u_1,\ldots,u_t\in\F_p
\]
and set
\[
  H(x)=\sum_{j=1}^t u_jF_j(x).
\]
Let
\[
  S=Z(R,F_1,\ldots,F_t),\qquad T_R=Z(R).
\]
Then \(S\subseteq Z(R,H)\) always, and
\[
  \Pr\bigl[Z(R,H)\ne S\bigr]\le \frac{E}{p}<\varepsilon.
\]
On the complementary event, \(H\) has exactly the same roots on \(Z(R)\) as the
full constraint list \(F_1,\ldots,F_t\).

The polynomial \(R=x^E-1\) is squarefree over \(\F_p\). Hence, for every
\(H\in\F_p[x]\),
\[
  \gcd(R,H)=
  \prod_{\beta\in Z(R),\,H(\beta)=0}(x-\beta),
\]
where the GCD is monic and the empty product is \(1\). Consequently, on the
successful compression event:
\begin{enumerate}
  \item if \(\Psi\) is unsatisfiable, then \(\gcd(R,H)=1\);
  \item if \(\Psi\) has exactly one satisfying assignment, then
  \(\gcd(R,H)=x-\alpha\) for the corresponding Boolean-domain root \(\alpha\),
  and in particular \(\deg\gcd(R,H)=1>0\).
  \end{enumerate}
\begin{proof}
If \(\beta\in S\), then every \(F_j(\beta)\) is zero, and hence
\[
  H(\beta)=\sum_{j=1}^t u_jF_j(\beta)=0.
\]
Thus \(S\subseteq Z(R,H)\) for every choice of the coefficients.

Now fix \(\alpha\in T_R\setminus S\). By definition of \(S\), at least one
\(F_j(\alpha)\) is nonzero. Condition on all coefficients except such a
coefficient \(u_j\). The equation \(H(\alpha)=0\) is then a nontrivial affine
linear equation in the uniform variable \(u_j\in\F_p\), so it holds with
probability \(1/p\). Taking the union bound over the \(E\) points of \(T_R\)
gives
\[
  \Pr\bigl[Z(R,H)\ne S\bigr]\le |T_R|/p=E/p<\varepsilon.
\]
On the complementary event, no point of \(T_R\setminus S\) is added as a root
of \(H\), while all points of \(S\) remain roots, so \(H\) and the full
constraint list have the same roots on \(Z(R)\).

It remains to relate these roots to the GCD. Since \(E\mid p-1\), the
characteristic \(p\) does not divide \(E\). Therefore
\[
  R'(x)=Ex^{E-1}
\]
does not vanish at a root of \(R=x^E-1\), and \(R\) is squarefree. All roots
of \(R\) in \(\F_p\) are the elements of \(T_R\), so the monic GCD of \(R\)
and \(H\) is exactly the product of the linear factors corresponding to the
common roots:
\[
  \gcd(R,H)=
  \prod_{\beta\in Z(R),\,H(\beta)=0}(x-\beta).
\]
If \(\Psi\) is unsatisfiable, Lemma~\ref{lem:encoding} gives \(S=\varnothing\);
on the successful compression event the product is empty and the GCD is \(1\).
If \(\Psi\) has exactly one satisfying assignment, Lemma~\ref{lem:encoding}
gives \(S=\{\alpha\}\); on the same event the product is \(x-\alpha\), whose
degree is \(1\).
\end{proof}
\end{lemma}

\begin{lemma}[Constructed instance satisfies the sparse target bounds]
\label{lem:instance-bounds}
For one trial of the reduction, set
\[
  f=H,\qquad g=R=x^E-1,\qquad n=E,\qquad D=2E.
\]
Then
\[
  g=x^n-1,\qquad n=E\le D=2E.
\]
The polynomial \(R\) has two terms. Each \(B_i\) has at most three terms after
collecting like terms, and each clause polynomial \(P_C\) has at most eight
monomial contributions before collection. Moreover
\[
  \deg B_i\le E,\qquad \deg P_C\le 2E.
\]
Therefore, after listing terms, collecting like exponents, and deleting zero
coefficients,
\[
  |\supp(H)|\le 3m+8k,\qquad \deg H\le 2E
\]
when \(H\ne 0\), while the zero polynomial has empty sparse support and
vacuous degree bound. All exponents and coefficients in the instance have
polynomial bit length in \(s\), and the sparse representation of \(H\) is
computable in polynomial bit time.
\begin{proof}
The identities \(g=x^n-1\) and \(n\le D\) follow immediately from the
definitions \(g=R=x^E-1\), \(n=E\), and \(D=2E\). The polynomial \(R\) has the
two displayed monomials \(x^E\) and \(-1\).

For each variable polynomial,
\[
  B_i(x)=x^{2E/\ell_i}-(1+z_i)x^{E/\ell_i}+z_i,
\]
so \(B_i\) has at most three terms after collecting like terms. Since
\(\ell_i\ge 2\),
\[
  \deg B_i=2E/\ell_i\le E.
\]
Each clause polynomial \(P_C\) is a product of three binomials, hence has at
most \(2^3=8\) monomial contributions before collection. If the three literals
involve indices \(a,b,c\), with repetitions allowed, then
\[
  \deg P_C\le E/\ell_a+E/\ell_b+E/\ell_c\le 3E/2\le 2E,
\]
because every \(\ell_i\ge 2\). Repeated literals do not increase this bound.

The polynomial \(H=\sum_j u_jF_j\) is a linear combination of the \(m\)
polynomials \(B_i\) and the \(k\) clause polynomials. Listing all their
monomial contributions gives at most \(3m+8k\) terms before collecting.
Collecting equal exponents and deleting zero coefficients can only decrease
the support size, so
\[
  |\supp(H)|\le 3m+8k.
\]
The same degree estimates give \(\deg H\le 2E\) whenever \(H\ne 0\); if
\(H=0\), the sparse support is empty and the degree condition is vacuous.

The exponents are bounded by \(2E\), so their binary lengths are
\(O(\log E)\). Lemma~\ref{lem:field-domain} gives
\(\log E\le \log p=\poly(s)\), and all coefficients lie in \(\F_p\). Since
\(m,k=\poly(s)\), the list of monomial contributions has polynomial size, and
the arithmetic for collecting equal exponents and removing zero coefficients
takes polynomial bit time.
\end{proof}
\end{lemma}

\begin{lemma}[Amplification gives a BPP SAT algorithm]
\label{lem:amplification}
Assume that there is a polynomial-time solver for the restricted sparse
coprimality problem of Theorem~\ref{thm-RUD-F}, or
equivalently for deciding whether \(\deg\gcd(f,x^n-1)>0\), on the constructed
instances. If the solver is randomized with bounded error, repeat it
independently so that each query has error at most \(\gamma\).

Let \(1/A(s)\) be the isolation success lower bound from
Lemma~\ref{lem:source}. Choose
\[
  r=\lceil 8A(s)\rceil,\qquad
  \varepsilon=\gamma=\frac{1}{100r}.
\]
Run \(r\) independent trials. In each trial, perform the isolation step,
construct the prime-field data and the sparse instance
\[
  (f,g,n,D)=(H,x^E-1,E,2E),
\]
query the amplified restricted solver, and accept the SAT input \(\Phi\) if
some trial reports ``not coprime'', equivalently
\(\deg\gcd(f,g)>0\).

If \(\Phi\) is unsatisfiable, then the false acceptance probability is at most
\[
  r(\varepsilon+\gamma)=1/50<1/3.
\]
If \(\Phi\) is satisfiable, then a trial accepts with probability at least
\[
  \frac{1}{2A(s)},
\]
and hence the probability that all \(r\) trials reject is at most
\[
  \left(1-\frac{1}{2A(s)}\right)^r
  \le \exp(-r/(2A(s)))\le e^{-4}<1/3.
\]
For strict worst-case BPP time, the Las Vegas field-construction subroutine may
be given a polynomial cutoff, with cutoff failure counted as an additional
per-trial bad event and the constants above decreased accordingly. Therefore a
polynomial-time solver for the restricted sparse coprimality problem would
imply \(\mathrm{SAT}\in\mathrm{BPP}\), and hence
\(\mathrm{NP}\subseteq\mathrm{BPP}\).
\begin{proof}
Run the stated \(r\) independent trials. In an unsatisfiable input, the
zero-preserving property from Lemma~\ref{lem:source} makes every isolated
formula \(\Psi\) unsatisfiable. By Lemma~\ref{lem:compression-GCD}, after a
successful compression event the constructed pair is coprime. Thus a false
accepting trial can occur only if compression fails or the amplified solver
answers incorrectly. The per-trial probability is at most
\(\varepsilon+\gamma\), and the union bound over all \(r\) trials gives
\[
  r(\varepsilon+\gamma)=r\left(\frac{1}{100r}+\frac{1}{100r}\right)
  =\frac1{50}<\frac13.
\]

Now suppose that \(\Phi\) is satisfiable. In one trial,
Lemma~\ref{lem:source} produces a uniquely satisfiable \(\Psi\) with
probability at least \(1/A(s)\). Conditional on this event, the compression is
successful and the amplified solver is correct with probability at least
\[
  1-\varepsilon-\gamma>1/2.
\]
On that joint event, Lemma~\ref{lem:compression-GCD} gives
\(\deg\gcd(H,x^E-1)=1>0\), and the solver reports ``not coprime.'' Therefore, a
single trial accepts with probability at least \(1/(2A(s))\). The probability
that all \(r=\lceil 8A(s)\rceil\) trials reject is at most
\[
  \left(1-\frac{1}{2A(s)}\right)^r
  \le \exp(-r/(2A(s)))\le e^{-4}<1/3.
\]

The construction in each trial has polynomial sparse bit size by
Lemma~\ref{lem:instance-bounds}, and the number of trials is polynomial in the
input length. If the field-construction subroutine is used in Las Vegas
expected time, this provides the expected-time accounting described above. For a
strict worst-case BPP algorithm, impose a polynomial cutoff on that Las Vegas
subroutine and count a cutoff as another bad per-trial event; reducing the
constants in the choices of \(\varepsilon\) and \(\gamma\) keeps the total
error below \(1/3\). Hence, the assumed solver yields a BPP algorithm for SAT,
and consequently \(\mathrm{NP}\subseteq\mathrm{BPP}\).
\end{proof}
\end{lemma}

\section{Conclusion}
\label{sec-conc}

In this paper, we show that output-sensitive sparse polynomial GCD computation over any computable field is NP-hard.
Taking the contrapositive, for sparse polynomials $f,g\in\F_q[x]$, under the standard complexity assumption
\(\mathrm{NP}\nsubseteq\mathrm{P}\), there is no algorithm that computes
\(\gcd(f,g)\) with bit complexity polynomial in \(\size(f)\), \(\size(g)\), and \(\size(\gcd(f,g))\).
%Equivalently, a positive answer to the question in this bit-complexity model would imply \(\mathrm{NP}\subseteq\mathrm{BPP}\).

We also show that the Roots of Unity Detection over finite fields is NP-hard under BPP reduction. Equivalently, determining whether a sparse polynomial \(f(x)\) is co-prime to $x^n-1$ is NP-hard under BPP reduction.

\section*{Acknowledgment}
This paper is supported by the Strategic Priority Research Program of CAS Grants XDA0480502 and XDA0480503, NSFC Grants 12288201 and 92270001.

\section*{Declaration of AI Use}
%The main results presented in this paper were obtained by an artificial intelligence agent system, \textit{MechMath Agent Team (MMAT)}~\cite{MMAT}, under human guidance. 
The main results presented in this paper are obtained through an interactive collaboration between the authors and an artificial intelligence agent system, \textit{MechMath Agent Team (MMAT)}~\cite{MMAT}. 
The authors assume full responsibility for the paper’s content.
%MMAT is a large language model driven agent designed to prove mathematical theorems expressed in both natural language and formal language in Lean. 
%
%Proofs of Theorems \ref{thm-m1} and \ref{thm-m2} are produced by  the natural language prover MechMath-NLProver of MMAT when the statements of the theorem are given as input.
%
%In this paper, the natural language prover MechMath-NLProver of MMAT is used.
%The proof for Theorem \ref{thm-m1} is automatically produced by MechMath-NLProver when the statement is given to it. 
%The proof for Theorem \ref{thm-m2} is produced by MechMath-NLProver in two steps: proofs of Theorem * and Theorem * are automatically generated by providing the statements to it.
%The correctness of the proofs is checked by the authors.

%The primary results of this paper were obtained using the AI agent \textit{MechMath Agent Team}~\footnote{\url{https://eonmath.github.io/mechmath}} under human guidance. 
%\textit{MechMath Agent Team} is a large language model (LLM) driven agent designed to prove mathematical theorems expressed in natural language by planning algebraic reductions, invoking external tools, and iteratively verifying intermediate claims. 
%In this study, the certificate verification phase was autonomously executed by \textit{MechMath Agent Team} via Python tool calls. 

\bibliographystyle{KLMM/klmm}   % 文献格式样式
\bibliography{Refs} 

@misc{MMAT,
  author = {Yichuan Cao and Ruichen Qiu and Junqi Liu and Jiaqi Wang and Dakai Guo and Ruyong Feng and Lihong Zhi and Xiao‑Shan Gao},
  title = {MechMath Agent Team: LLM Driven Agents for Mathematical Research},
  year = {2026},
  eprint = {2607.04394},
  archiveprefix = {arXiv}
}

@misc{qiu2026outputsensitive,
  author = {Ruichen Qiu and Yichuan Cao and Qiao-Long Huang and Ruyong Feng and Xiao-Shan Gao},
  title = {Output-sensitive Sparse Polynomial GCD over Finite Fields is NP-hard},
  year = {2026},
  eprint = {2606.12144v1},
  archiveprefix = {arXiv}
}

@inproceedings{Roche2018,
    author = {Roche, Daniel S.},
    title = {What Can (and Can't) we Do with Sparse Polynomials?},
    year = {2018},
    isbn = {9781450355506},
    publisher = {Association for Computing Machinery},
    address = {New York, NY, USA},
    url = {https://doi.org/10.1145/3208976.3209027},
    doi = {10.1145/3208976.3209027},
    booktitle = {Proceedings of the 2018 ACM International Symposium on Symbolic and Algebraic Computation},
    pages = {25–30},
    numpages = {6},
    location = {New York, NY, USA},
    series = {ISSAC '18}
}

@article{huang2026sparse,
  author = {Qiao-Long Huang and Xiao-Shan Gao},
  title = {Sparse Polynomial GCD Algorithms Asymptotically Linear in All Fundamental Parameters},
  year = {2026},
  journal = {arXiv2609.08074}
  }

@article{van2021sparse,
    author = {van der Hoeven, Joris and Lecerf, Gr\'{e}goire},
    title = {On sparse interpolation of rational functions and gcds},
    year = {2021},
    volume = {55},
    number = {1},
    journal = {ACM Commun. Comput. Algebra},
    pages = {1–12},
    numpages = {12}
}

@misc{hong2023computing,
      title={Computing greatest common divisor of several parametric univariate polynomials via generalized subresultant polynomials}, 
      author={Hoon Hong and Jing Yang},
      year={2024},
      eprint={2401.00408},
      archivePrefix={arXiv},
      primaryClass={cs.SC},
      url={https://arxiv.org/abs/2401.00408}, 
}

@article{van2025optimizing,
  author = {van der Hoeven, Joris},
  title = {Optimizing the half-gcd algorithm},
  journal = {Applicable Algebra in Engineering, Communication and Computing},
  year = {2025},
  month = may,
  url = {https://doi.org/10.1007/s00200-025-00690-w}
}

@article{bernstein2019fast,
  title={Fast constant-time gcd computation and modular inversion},
  author={Bernstein, Daniel J and Yang, Bo-Yin},
  journal={IACR transactions on cryptographic hardware and embedded systems},
  pages={340--398},
  year={2019}
}

@inproceedings{FGS2008, 
    place={Cambridge}, 
    series={London Mathematical Society Lecture Note Series}, 
    title={Irreducibility and greatest common divisor algorithms for sparse polynomials}, 
    booktitle={Number Theory and Polynomials}, 
    publisher={Cambridge University Press}, 
    author={Filaseta, Michael and Granville, Andrew and Schinzel, Andrzej}, 
    editor={McKee, James and Smyth, ChrisEditors}, 
    year={2008}, 
    pages={155–176}
}

@inproceedings{kaltofen2005complexity,
  title={On the Complexity of Factoring Bivariate Supersparse (Lacunary) Polynomials},
  author={Kaltofen, Erich and Koiran, Pascal},
  booktitle={Proceedings of the 2005 International Symposium on Symbolic and Algebraic Computation (ISSAC '05)},
  year={2005},
  pages={208--215},
  publisher={ACM}
}

@inproceedings{Gathen1993CC,
  title={Counting curves and their projections},
  author={Gathen, Joachim von zur and Karpinski, Marek and Shparlinski, Igor},
  booktitle={Proceedings of the twenty-fifth annual ACM symposium on Theory of Computing},
  series={STOC'93},
  year={1993},
  pages={805--812},
  publisher={ACM}
  }

@inproceedings{Karpinski1999CH,
  title={On the computational hardness of testing square-freeness of sparse polynomials},
  author={Marek Karpinski and Igor Shparlinski},
  booktitle={Proc. AAECC-13 (Heidelberg, Germany, 1999)},
  series={Vol. 1719 of LNCS},
  year={1999},
  pages={492--497},
  publisher={Springer Verlag}
  }

@inproceedings{VV1986,
    author = {Valiant, L.G. and Vazirani, V.V.},
    title = {NP is as easy as detecting unique solutions},
    year = {1985},
    booktitle = {Proceedings of the Seventeenth Annual ACM Symposium on Theory of Computing},
    pages = {458–463},
    numpages = {6},
    location = {Providence, Rhode Island, USA},
    series = {STOC '85}
}

@inproceedings{Moses1973EZ,
    author = {Moses, Joel and Yun, David Y. Y.},
    title = {The EZ GCD algorithm},
    year = {1973},
    publisher = {Association for Computing Machinery},
    address = {New York, NY, USA},
    booktitle = {Proceedings of the ACM Annual Conference},
    pages = {159–166},
    numpages = {8},
    location = {Atlanta, Georgia, USA},
    series = {ACM '73}
}

@inproceedings{zippel1979probabilistic,
  author       = {Richard Zippel},
  title        = {Probabilistic Algorithms for Sparse Polynomials},
  booktitle    = {Symbolic and Algebraic Computation, EUROSAM '79, An International Symposium on Symbolic and Algebraic Computation},
  editor       = {Edward W. Ng},
  series       = {Lecture Notes in Computer Science},
  volume       = {72},
  pages        = {216--226},
  publisher    = {Springer},
  year         = {1979}
  }

@article{collins1967subresultants,
    author = {Collins, George E.},
    title = {Subresultants and Reduced Polynomial Remainder Sequences},
    year = {1967},
    issue_date = {Jan. 1967},
    publisher = {Association for Computing Machinery},
    address = {New York, NY, USA},
    volume = {14},
    number = {1},
    issn = {0004-5411},
    url = {https://doi.org/10.1145/321371.321381},
    doi = {10.1145/321371.321381},
    journal = {J. ACM},
    month = jan,
    pages = {128–142},
    numpages = {15}
}

@article{brown1971euclid,
    author = {Brown, W. S.},
    title = {On Euclid's Algorithm and the Computation of Polynomial Greatest Common Divisors},
    journal={Journal of the ACM},
    year = {1971},
    volume = {18},
    number = {4},
    pages = {478–504}
}

@article{brown1971subresultants,
  title={On Euclid's Algorithm and the Theory of Subresuhants},
  author={Brown, W. S. and Traub, J. F.},
  journal={Journal of the ACM},
  volume={18},
  number={4},
  pages={505--514},
  year={1971},
  publisher={ACM New York, NY, USA}
}

@book{ModernCA,
  title={Modern Computer Algebra},
  author={von zur Gathen, Joachim and Gerhard, J{\"u}rgen},
  year={2013},
  publisher={Cambridge University Press}
}

@article{GCDHuangGao2025,
  author = {Qiao-Long Huang and Xiao-Shan Gao},
  title = {Bit Complexity of Polynomial GCD on Sparse Representation},
  journal = {Mathematics of Computation},
  volume = {95},
  number = {357},
  pages = {389-413},  
  year = {2026}
}

@inproceedings{hu2016fast,
    author = {Hu, Jiaxiong and Monagan, Michael},
    title = {A Fast Parallel Sparse Polynomial GCD Algorithm},
    year = {2016},
    publisher = {Association for Computing Machinery},
    address = {New York, NY, USA},
    booktitle = {Proceedings of the 2016 ACM International Symposium on Symbolic and Algebraic Computation},
    pages = {271–278},
    numpages = {8},
    location = {Waterloo, ON, Canada},
    series = {ISSAC '16}
}

@article{kaltofen1990computing,
    title = {Computing with polynomials given by black boxes for their evaluations: Greatest common divisors, factorization, separation of numerators and denominators},
    journal = {Journal of Symbolic Computation},
    volume = {9},
    number = {3},
    pages = {301-320},
    year = {1990},
    author = {Erich Kaltofen and Barry M. Trager}
}

@article{GCD-Kaltofen1988,
    author = {Kaltofen, Erich},
    title = {Greatest common divisors of polynomials given by straight-line programs},
    year = {1988},
    volume = {35},
    number = {1},
    journal = {J. ACM},
    pages = {231–264}
    }

@article{plaisted1984,
    title = {New NP-hard and NP-complete polynomial and integer divisibility problems},
    journal = {Theoretical Computer Science},
    volume = {31},
    number = {1},
    pages = {125-138},
    year = {1984},
    author = {David A. Plaisted},
}

@article{plaisted1977sparse,
    title = {Sparse complex polynomials and polynomial reducibility},
    journal = {Journal of Computer and System Sciences},
    volume = {14},
    number = {2},
    pages = {210-221},
    year = {1977},
    author = {David Alan Plaisted}
}

@inproceedings{GCD-Schinzel2003, 
    title={On the Greatest Common Divisor of Two Univariate Polynomials, I}, 
    booktitle={A Panorama of Number Theory or The View from Baker’s Garden}, 
    publisher={Cambridge University Press}, 
    author={Schinzel, A.}, 
    editor={Wüstholz, GisbertEditor}, 
    year={2002}, 
    place={Cambridge}, 
    pages={337–352}
}

@inproceedings{SparsityChallenges,
  author={Davenport, James Harold and Carette, Jacques},
  booktitle={2009 11th International Symposium on Symbolic and Numeric Algorithms for Scientific Computing}, 
  title={The Sparsity Challenges}, 
  year={2009},
  volume={},
  number={},
  pages={3-7}
}

\end{document}